\documentclass[letterpaper, 10 pt, conference]{ieeeconf}  % Comment this line out if you need a4paper

\usepackage{amsmath,amssymb,amsfonts}
\usepackage{graphicx} 
\usepackage{mathtools}
\usepackage{algorithm}
\usepackage{algpseudocode}
\usepackage{graphicx}
\usepackage{textcomp}
\usepackage{caption}
\usepackage{subcaption}
\usepackage[normalem]{ulem}
\usepackage{balance}
\usepackage[
backend=biber,
style=ieee,
citestyle=ieee,
sorting=none
]{biblatex}

\newtheorem{theorem}{Theorem}
\newtheorem{lemma}{Lemma} 
\newtheorem{assumption}{Assumption}
\newtheorem{definition}{Definition} 
\newtheorem{proposition}{Proposition} 
\newtheorem{remark}{Remark}

\newcommand{\R}{\mathbb{R}} 
\newcommand{\node}{\mathcal{V}}     % Vertex set
\newcommand{\edge}{\mathcal{E}}     % Edge set
\newcommand{\gd}{\mathcal{V}_G}     % set of Good agents
\newcommand{\nbrTdN}[2]{\mathcal{X}_{#1} (#2)} 
\newcommand{\nbrAdv}[2]{\overline{\mathcal{X}}_{#1} (#2)}

\newcommand{\xMax}[1]{x_{\text{max}} (#1)}
\newcommand{\xMin}[1]{x_{\text{min}} (#1)}

\newcommand{\xTd}[1]{x_{\text{td}} (#1)} 
\newcommand{\rTd}[1]{r_{\text{td}} (#1)}  
\newcommand{\rBarTd}[1]{\Bar{r}_{\text{td}} (#1)} 
\newcommand{\xBarTd}[1]{\Bar{x}_{\text{td}} (#1)}

\newcommand{\one}{\mathbf{1}} 

\newcommand{\delR}[1]{\Delta R (#1)}

\title{\LARGE \bf
Resilient Dynamic Average Consensus based on Trusted agents
}

\author{Shamik Bhattacharyya, Rachel Kalpana Kalaimani% <-this % stops a space
}

\begin{document}

\balance

\maketitle
\thispagestyle{empty}
\pagestyle{empty}

%%%%%%%%%%%%%%%%%%%%%%%%%%%%%%%%%%%%%%%%%%%%%%%%%%%%%%%%%%%%%%%%%%%%%%%%%%%%%%%%
\begin{abstract}
In this paper we address discrete-time dynamic average consensus (DAC) of a multi-agent system in the presence of adversarial attacks. The adversarial attack is considered to be of Byzantine type, which compromises the computation capabilities of the agent and sends arbitrary false data to its neighbours. We assume few of the agents cannot be compromised to adveraries which we term as trusted agents. We first formally define resilient DAC in the presence of Byzantine adversaries. Then we propose our novel Resilient Dynamic Average Consensus (ResDAC) algorithm that ensures the trusted and ordinary agents achieve resilient DAC in the presence of adversarial agents. The only requirements are that of the trusted agents forming a connected dominating set, and the first order differences of the reference signals being bounded. We do not impose any restriction on the tolerable number of adversarial agents that can be present in the network. We also do not restrict the reference signals to be bounded. Finally we provide numerical simulations to illustrate the effectiveness of the proposed ResDAC algorithm.  
\end{abstract}

%%%%%%%%%%%%%%%%%%%%%%%%%%%%%%%%%%%%%%%%%%%%%%%%%%%%%%%%%%%%%%%%%%%%%%%%%%%%%%%%
\section{INTRODUCTION}
The problem of \emph{dynamic average consensus} (DAC) was first studied in \cite{DAC_Murray_IFAC05}. It considers a network of agents where each agent locally measures a time-varying network quantity to obtain a local reference signal. The objective is to develop a distributed algorithm that will allow each agent to asymptotically track the average of the time-varying reference signals. \cite{DAC_Mag19} provides a comprehensive study of the available literature on DAC.

The expanding range of applications of distributed consensus based algorithms in multiagent systems has also unfortunately attracted a rise in cases of cyber attacks \cite{AttckSrvey_IndInfo22}. Such attacks on networked multi-agent systems try to disrupt the proper functioning of the distributed protocols and prevent them from achieving the common objective. The type of attacks can be broadly classified into three categories based on the main functions of the agents - sensing, communication and computation. The attack on sensors like \emph{false data injection} \cite{ResDAC} try to corrupt the measurement data received by the agents leading to faulty operations. The second type of attacks, like \emph{Denial of Service} \cite{ResCon_DoS}, tamper with the communication capabilities to prevent the regular exchange of information among the neighbouring agents. In both the above type of attacks, the agents are usually still able to carry on with their computation based on the designed distributed protocols. This is where the third category of attacks, like \emph{Byzantine}, \emph{malicious} attacks, focus on. They tend to corrupt the computation capabilities of the agents and send arbitrary false information to its neighbours. The agents under such attacks cannot be expected to follow the designed protocol and thus the common objective of the network requires to redefined based on the non-adversarial agents. 

With such wide variety of adversarial attacks, it is imperative to develop suitable resilient distributed algorithms that would ensure common objectives are still achieved in the presence of attacks. There exists appreciable literature on resilient (static) consensus algorithms \cite{ResCon_Sundaram_13}, \cite{ResCon_Ishii_20} where the agents asymptotically converge to a common fixed value in the presence of adversarial attacks. The commonly used methods to design such algorithms are the \emph{mean subsequence reduced (MSR)} approach \cite{ResCon_Sundaram_13}, and the \emph{trusted agents} based approach \cite{ResCon_Trstd_14}. Resilient (static) consensus algorithms have also inspired the development of other distributed consensus-based algorithms resilient to adversarial attacks, like resilient distributed estimation \cite{REWB}, resilient distributed optimization \cite{ResDO_Trstd_TAC20}, etc. When it comes to resilient dynamic average consensus, a recent work \cite{ResDAC} addressed the problem considering attacks on sensor nodes and the time-varying reference signal to be bounded. They show that in the presence of uniformly bounded false-data injection attacks, the agents achieve dynamic average consensus within a desired tolerance.

To the best of our knowledge, apart from \cite{ResDAC} there is no other work addressing dynamic consensus in the presence of adversarial attacks till date. In this paper we address resilient dynamic consensus in the presence of Byzantine adversaries which is a different scope of adversarial attack compared to false-data injection attacks considered in \cite{ResDAC}, as explained before. Our main contributions in this paper are listed below
\begin{itemize}
    \item We formally define resilient DAC based on trusted agents (Definition 4). We first explain why in the presence of Byzantine or malicious adversaries,  it is not feasible to track the average of the reference signals of all the agents in the network. Then we present our definition motivated by the work on trusted agent based resilient algorithms.

    \item We propose the novel ResDAC algorithm that provides state update laws for the trusted and ordinary agents to ensure they track the average of the reference signals of the trusted agents.

    \item We show that when the trusted agents induce a connected dominating set and the first order differences of the local reference signals are bounded, the proposed ResDAC algorithm ensures that the trusted and ordinary agents achieve resilient dynamic average consensus (Theorem 1). We do not enforce any limitation on the number of adversarial agents that may be present in the network. 

\end{itemize}

\textit{Notations}. $\R$ denotes the set of \emph{real} numbers, and $\R^N$ represents the $N$-dimensional Eucledian space. $|\mathcal{S}|$ denotes the \textit{cardinality} of any given set $\mathcal{S}$. $\one$ denotes a vector of all $1$s, $\one = (1,1,\hdots,1)$ of appropriate dimensions. For a real-valued vector $v$, $v^T$ denotes the \emph{transpose} of the vector. Similarly for a real-valued matrix $M$, $M^T$ denotes the \emph{transpose} of the matrix, and $[M]_{i:}$ denotes its $i$-th row. $[M]_{i:j,k:l}$ denotes the sub-matrix of $M$ consisting of the entries lying within its $i$ to $j$-th rows and $k$ to $l$-th columns. 

\section{Problem Formulation} \label{sec:probFormulatn} 
 \subsection{System Model} \label{sec:pF_sysModel}
We consider a network of $M$ agents represented by the set $\node = \{1, \hdots, M \}$. The agents interact over a communication topology represented by an undirected graph $\Gamma = (\node,\edge)$. 
The agents are of 3 types - trusted, ordinary and adversarial. The \emph{trusted} agents, represented by set $\node_T$, maintain a high level of security such that it can never be compromised by an attacker. They form a kind of secure backbone for the network. The ordinary agents, on the other hand, do not have any such additional security measures and are thus vulnerable to attacks. The set of ordinary agents, denoted by $\node_O$, represents the agents which are neither trusted nor under adversarial attack. The adversarial agents are the ones under attack by an adversary, and are represented by the set $\node_A$. Let the number of trusted, ordinary and adversarial agents be denoted by $m_T$, $m_O$ and $m_A$ respectively. Note that $m_T + m_O + m_A = M$. 
We refer to the set of trusted and ordinary agents together as the set of good agents $\gd = \node_T \cup \node_O$ with $|\gd| = m_T + m_O = N$. Without loss of generalization, let us represent the set of agents as $\node = \{1, \hdots, m_T, m_T+1, \hdots, N, N+1, \hdots, M\}$, thus arranging the agents starting with the trusted, followed by the ordinary and finally the adversarial ones. 

Next we define a connectivity property of the graph from \cite{ResDO_Trstd_TAC20}.
\begin{definition}[Connected Dominating Set (CDS)]
A set $\mathcal{S}$ of graph $\Gamma = (\node,\edge)$ is a CDS if  
\begin{itemize}
    \item all nodes belonging to $\mathcal{S}$ form a connected graph, and
    \item each node which does not belong to $\mathcal{S}$ has at least one neighbour in $\mathcal{S}$.
\end{itemize} 
\end{definition}
Now we state our first assumption based on the trusted agents and CDS which we will later use for our main result.
\begin{assumption} \label{asmpn:td_cds}
The set of trusted agents $\node_T$ induce a CDS of $\Gamma = (\node,\edge)$.
\end{assumption}
This assumption can be seen as investing in higher security for a subset of the agents in a network, the trusted agents, to ensure they function properly in the presence of any number of adversarial agents. Compared to the popular graph robustness approach, used in many resilient consensus based algorithm, it is shown in \cite{TdNodes_Xeno_TCNS18} that controlling the number and location of trusted nodes provides any desired network connectivity and robustness that too without the need of adding extra links among agents. The condition of each ordinary node being connected to at least one trusted neighbour is comparable to the commonly used robustness criteria that requires each non-adversarial agent to have at least one non-adversarial neighbour. 

Next we define a subgraph of $\Gamma$, which we later use in Lemma \ref{lem:eqModel} to develop the equivalent model representation of our proposed algorithm.
\begin{definition}
    $\Gamma_G = (\node_G,\edge_G)$ is a subgraph of $\Gamma$, where $\node_G = \node_T \cup \node_O$, and $\edge_G \in \edge$ consists of all the edges interconnecting the trusted agents, and the edges incoming to the ordinary agents from their trusted neighbours.
\end{definition} 

 \subsection{Resilient Dynamic Average Consensus} \label{sec:pF_RDC}
 Each agent synchronously measures a local continuous physical process $r_i(t) : \R \rightarrow \R$, $t \in \mathbb{N}$. We refer to $r_i(t)$ as the \emph{reference signal} of agent $i$ at time instant $t$. The collective aim of the agents is to arrive at consensus over the average of the reference signals in a distributed manner. For this, each agent maintains its own local variable as $x_i(t)$. At every time instant $t$, each agent : 
 \begin{itemize}
     \item obtains the reference signal $r_i(t)$ 
     \item communicates its own state $x_i(t)$ to its neighbours 
     \item receives neighbours' state values $x_j(t), j \in \mathcal{N}_i$ 
     \item updates its own state as 
     \begin{equation*}
         x_i(t+1) = f_i:\{x_i(t),r_i(t),r_i(t-1),x_j(t)|j\in \mathcal{N}_i\} \rightarrow \R
     \end{equation*}
 \end{itemize} 
Now some of the agents are under adversarial attack. We first define the type of adversarial attacks we consider in this paper.
\begin{definition}[Byzantine and Malicious Adversary] 
    An agent $i \in \node_A$ is said to be \emph{Byzantine} if it updates its state using some arbitrary function $f'_i$ and sends different values to different neighbors. The agent is \emph{malicious} if it updates its state using some arbitrary function $f'_i$ and sends the same value to all of its neighbours. 
\end{definition} 
As malicious adversary is a special case of the Byzantine type, we henceforth consider only Byzantine adversarial agents. Now with such adversarial agents sharing arbitrary and manipulated values to its neighbours, it would not be possible to obtain the true information regarding the value of their reference signals. So tracking the average of the reference signals of all the agents in a distributed manner is not feasible. Our focus here is to ensure resilience to adversarial agents based on trusted agents. In studies of resilient consensus protocols based on trusted agents, the results ensure that the trusted and ordinary agents achieve consensus over some value based on the trusted agents. In \cite{ResCon_Trstd_14}, resilient consensus is achieved by trusted and ordinary agents over the limit value of trusted agents, while in \cite{ResDO_Trstd_TAC20} the final solution of resilient distributed optimization lies within the convex hull of minimizers of the local cost function of trusted agents. Motivated by the previous results, we aim to ensure that the trusted and ordinary agents are able to track the average of the reference signals of the trusted agents in the presence of adversaries, which is formally defined below.
%So for a network with Byzantine adversarial agents, our resilience design based on trusted agents is focused on ensuring that the trusted and ordinary agents achieve resilient dynamic average consensus, which is formally defined below. 
\begin{definition}[Resilient Dynamic Average Consensus]
    In the presence of Byzantine adversarial agents, the good agents are said to achieve resilient dynamic average consensus based on trusted agents, if all the good agents track the average of the reference signals of the trusted agents, i.e. $\forall i \in \gd$, 
    \begin{itemize}
        \item[] $ \lim_{t \rightarrow \infty} |x_i(t) - \rBarTd{t-1}| \leq \epsilon$ for some $ \epsilon \geq 0 $.
    \end{itemize}
\end{definition}

The challenge then is to design an algorithm that provides a suitable $f_i$ for the good agents to update their states in a way that they are able to achieve resilient dynamic average consensus in the presence of adversarial agents within the network.

\section{Results}  \label{mainRslt} 
 \subsection{Algorithm} \label{sec:mR_algo} 
 \begin{algorithm} 
\caption{ResDAC} 
\label{alg:rdc}
 \textbf{Given} : $\mathcal{T}_i$ and $r_i(0) \in \R$ for each $i \in \gd$, $m_T$
\\ \textbf{Initialize} : $x_i(1) \in \R$ for each $i \in \gd$
\\ \textbf{for} $t = 1,2,\hdots $ \textbf{do} 
\begin{itemize} 
    \item[] \textbf{for} each trusted agent $i \in \node_T $ \textbf{do} 
\begin{itemize}
    \item measure $r_i(t)$ and calculate $\Delta r_i(t) = r_i(t) - r_i(t-1)$
    \item collect states of trusted neighbours $\mathcal{T}_i$ 
    \item update $x_i(t+1)$ as 
     \begin{equation} \label{eq:upLaw_td}
         x_i(t+1) = \sum_{j \in \mathcal{T}_i \cup \{i\}} v_{ij}(t) x_j(t) + \Delta r_i(t)
     \end{equation}     
\end{itemize}    
    \item[] \textbf{for} each ordinary agent $i \in \node_O $ \textbf{do}
\begin{itemize} 
    \item measure $r_i(t)$ and calculate $\Delta r_i(t) = r_i(t) - r_i(t-1)$
    \item collect states of trusted neighbours $\mathcal{T}_i$ in $\mathcal{S}_i(t)$   
    \item sort the values in $\mathcal{S}_i(t) \cup \{x_i(t)\}$ and store the min. and max. values in $x_i^{\text{min}}(t)$ and $x_i^{\text{max}}(t)$ respectively 
    \item create $\mathcal{U}_i(t) = \{ j| x_j(t) \in [x_i^{\text{min}}(t), x_i^{\text{max}}(t)] , j \in \mathcal{N}_i \cup \{i\} \}$ 
    \item update $x_i(t+1)$ as 
     \begin{equation} \label{eq:upLaw_ord}
         x_i(t+1) = \sum_{j \in \mathcal{U}_i(t)} \frac{1}{|\mathcal{U}_i(t)|} x_j(t) + \Delta r_i(t)
     \end{equation}     
\end{itemize} 
\textbf{end for}
\end{itemize}
\textbf{end for}
\\ \textbf{Output : } $x_i(t)$ for all $i \in \gd$
\end{algorithm}
 In this section we introduce our \emph{Resilient Dynamic Average Consensus} (ResDAC) algorithm. The ResDAC algorithm is designed for the good agents to achieve resilient dynamic average consensus. We consider that each good agent $i$ knows its set of trusted neighbours, $\mathcal{T}_i = \{ j|j \in \mathcal{N}_i \cap \node_T \}$. At the start of every iteration, each good agent $i$ measures its local reference signal $r_i(t)$ and then calculates the change in the reference signal as $\Delta r_i(t) = r_i(t) - r_i(t-1)$. Our algorithm has two distinct parts - one for the state update of the trusted agents, and the other for the update process followed by the ordinary agents. Let us first consider the case of a trusted agent $i \in \node_T$. After collecting the states of its trusted neighbours $\mathcal{T}_i$, agent $i$ updates its state following the update law in \eqref{eq:upLaw_td}. The weight $v_{ij}(t)$ is defined as $v_ij(t) = 1/m_T$ for $ j \in \mathcal{T}_i$, $1 - |\mathcal{T}_i|/m_T$ for $j=i$, and $0$ otherwise. This choice of weights for the trusted agents' update ensures that the exact average of the reference signals of the trusted agents is tracked by all the good agents in the presence of adversarial neighbours. This is later illustrated through the equivalent model in Lemma \ref{lem:eqModel}, and in the proof of Theorem \ref{thm:consensus}. Now considering the case of an ordinary agent $i \in \node_O$. After collecting the states of its trusted neighbours in $\mathcal{S}_i(t) = {x_j(t)|j \in \mathcal{T}_i}$, it sorts the values in $\mathcal{S}_i(t) \cup \{x_i(t)\}$. From the sorted list, the minimum value is stored as $x_i^{\text{min}}(t)$ and the maximum value is stored as $x_i^{\text{max}}(t)$. Then among all its neighbours, only those are selected whose state value falls in the range $[x_i^{\text{max}}(t), x_i^{\text{min}}(t)]$ and are enlisted in the set $\mathcal{U}_i(t) = \{ j| x_j(t) \in [x_i^{\text{min}}(t), x_i^{\text{max}}(t)] , j \in \mathcal{N}_i \cup \{i\} \}$. Finally agent $i$ follows the update law \eqref{eq:upLaw_ord} to update its state. 

 \subsection{Results} \label{sec:mR_mainRslt}
 Let us first define a few terms, which we then use to state our assumption to relatively bound the first order differences of the reference inputs.
 \begin{equation*}
     \Delta r_{\text{min}}(t) = \min_{i \in \node_G} \Delta r_i(t), 
     \Delta r_{\text{max}}(t) = \max_{i \in \node_G} \Delta r_i(t)
 \end{equation*}
\begin{assumption}  \label{asmpn:delR_bound}
    The first-order differences of the reference signals is relatively bounded. Specifically, there exists a time-invariant constant $\theta > 0$ such that 
    \begin{equation} \label{eq:delR_bound}
        \delR{t} := \Delta r_{\text{max}}(t) - \Delta r_{\text{min}}(t) \leq \theta, \forall t \geq 0
    \end{equation}
\end{assumption}
The above assumption is to ensure that the reference signals are not varying too fast. In a distributed approach it takes time for the information of an agent to percolate to all other agents in the network. So it is fairly reasonable to assume that the local signals, contributing to the common time-varying parameter being tracked by the agents, vary slowly enough to allow the agents to track the desired value.

Let $\xTd{t} = [x_1(t), \hdots, x_{m_T}(t)]^T$ and $x(t) = [\xTd{t}^T, x_{m_T+1}(t), \hdots, x_{N}(t)]^T$ be the vectors representing the states of the trusted and good agents respectively. Also let $\rTd{t} = [r_1(t), \hdots, r_{m_T}(t)]^T$ and $r(t) = [\rTd{t}^T, r_{m_T+1}(t), \hdots, r_{N}(t)]^T$ be the vectors representing the reference signals of the trusted and good agents respectively. Then $\Delta \rTd{t} = [\Delta r_1(t), \hdots, \Delta r_{m_T}(t)]^T$ and $\Delta r(t) = [\Delta \rTd{t}^T, \Delta r_{m_T+1}(t), \hdots, \Delta r_{N}(t)]^T$. Now we establish the existence of a transition matrix for the update of the states of the good agents based on our ResDAC algorithm. The properties of the transition matrix are also presented which are crucial for establishing our main result.
\begin{lemma} \label{lem:eqModel}
    Consider an undirected graph $\Gamma$ where trusted agents satisfy Assumption \ref{asmpn:td_cds}, and the corresponding subgraph $\Gamma_G$. Then for the ResDAC algorithm, there exists $W(t) \in \R^{N \times N}$ for all $t$ such that 
    \begin{equation} \label{eq:upLaw_vctr}
        x(t+1) = W(t) x(t) + \Delta r(t)
    \end{equation} 
    where $W(t) = [W_{ij}(t)]$ has the following properties :
    \begin{enumerate}
        \item[A1)] $W(t)$ is row stochastic, i.e., $W(t) \one = \one$ ; 
        \item[A2)] $W_{ij}(t) \neq 0 \iff (j,i) \in \edge_G \cup \{(i,i)\}$, $\forall i \in \gd$ ; 
        \item[A3)] $\forall$ $W_{ij}(t) \neq 0$, $W_{ij}(t) \geq \alpha = 1/(1 + d^{\text{in}}_{\text{max}})$, $\forall i \in \gd$.
    \end{enumerate}
    Moreover, let the upper left square block of $W(t)$ be denoted by sub-matrix $\Hat{W}(t) = [W_{ij}(t)]_{1:m_T,1:m_T}$. Then $\Hat{W}(t)$ satisfies the following property :
    \begin{enumerate}
        \item[B1)] $\Hat{W}(t)$ is doubly stochastic, i.e., $\Hat{W}(t) \one = \one$ and $\one^T \Hat{W}(t) = \one^T$
    \end{enumerate}
\end{lemma}
Now we present the main result of our work on resilient dynamic average consensus based on trused agents using the proposed ResDAC algorithm.
\begin{theorem} \label{thm:consensus}
    Consider an undirected graph $\Gamma$ where trusted agents satisfy Assumption \ref{asmpn:td_cds}. Also consider that the reference signals of all good agents satisfy Assumption \ref{asmpn:delR_bound}. Then the ResDAC algorithm ensures that all the ordinary and trusted agents achieve resilient dynamic average consensus, in the presence of Byzantine adversaries. In particular, 
    \begin{equation} \label{eq:thm}
        \lim_{t \rightarrow \infty} |x_i(t) - \rBarTd{t-1}| \leq \epsilon \text{, for all } i \in \gd  
    \end{equation} 
    where $\rBarTd{t} = \frac{1}{m_T} \rTd{t}$ and $\epsilon = \theta (N-1) (1 + \alpha^{-\frac{1}{2}N(N+1) + 1}) + |\xBarTd{1} - \rBarTd{0}|$. 
\end{theorem} 
\begin{remark}
    In Theorem 1 we can see that there is no bound on the tolerable number of adversarial agents present in the network. This means that even in the presence of any large number of adversarial agents, the ResDAC algorithm will ensure resilient dynamic average consensus of the all the good agents. The only condition required is that of the trusted agents admitting a CDS.   
\end{remark} 
\begin{remark}
    We can further infer from Theorem 1 that the ResDAC algorithm ensures proper tracking of the average of the reference signals of the trusted agents by all the good agents as long as the reference signals are slowly varying. This is specified by the bound on their first order differences. This allows for the application of the algorithm to track the time-varying average of various types of reference signals, like asymptotically decaying, sinusoidally varying and even ramp type signals.
\end{remark}  
Before starting with the proof of Theorem \ref{thm:consensus}, we first present some results which would be used later in the proof. Let for every $s$, fix some $k \in \node_T$ and define $\mathcal{D}_0 = \{ k \}$. Let $\mathcal{D}_1 \subset \node_G \backslash \{k\}$ denote the set of good agents with which agent $k$ communicates at time $s$. $\mathcal{D}_1$ is non-empty by Assumption \ref{asmpn:td_cds}. Using induction we have a set $\mathcal{D}_{l+1} \subset \node_G \backslash \mathcal{D}_1 \cup \hdots \cup \mathcal{D}_l$ consisting of those agents to which some $i \in \mathcal{D}_1 \cup \hdots \cup \mathcal{D}_l$ communicates at time step $s+l$. $\mathcal{D}_{l+1}$ is non-empty by Assumption \ref{asmpn:td_cds}, provided $\node_G \backslash \mathcal{D}_1 \cup \hdots \cup \mathcal{D}_l$ is non-empty. Thus $\mathcal{D}_1,\hdots, \mathcal{D_L}$ is a partition of $\node_G$ for some $\mathcal{L} \leq N - 1$. Now we present a result inspired from \cite[Lemma 3.1]{DT_DAC}.
\begin{proposition} \label{prop:xMax_xMin}
    Consider the RDC algorithm and suppose Assumption \ref{asmpn:td_cds} holds. Let $s \geq 0$ and $k \in \node_T$ be fixed, and consider the associated $\mathcal{D}_1, \hdots, \mathcal{D_L}$. Then for every $l \in \{1, \hdots, \mathcal{L} \}$, there exists a real number $\eta_l > 0$ such that for every integer $p \in [l,\mathcal{L}]$, and for $i \in \mathcal{D}_l$, it holds for $t=s+p$ 
    \begin{equation} \label{eq:xi_lwrBnd}
        x_i(t) \geq \xMin{s} + \sum_{q=0}^{p-1} \Delta r_{\text{min}}(s+q) + \eta_l (x_k(s) - \xMin{s})
    \end{equation}
    \begin{equation} \label{eq:xi_uprBnd}
        x_i(t) \leq \xMax{s} + \sum_{q=0}^{p-1} \Delta r_{\text{max}}(s+q) - \eta_l (\xMax{s} - x_k(s))
    \end{equation}
\end{proposition}
The proof of the above proposition is presented in Appendix. Now we proceed to prove Theorem \ref{thm:consensus}.

 \begin{proof}[Proof of Theorem \ref{thm:consensus}]
    Let $\eta = \alpha^{\frac{1}{2}N(N+1) - 1}$. Then for any $l \in [0,1, \hdots , N-1]$, $\eta \leq \eta_l$. From \eqref{eq:xi_lwrBnd}, by replacing $t$ and $s$ with $t_1 = t+ \mathcal{L}$ and $t$ respectively, for every $t \geq 0$ we get 
    \begin{align*}
       \xMin{t_1} &= \min_{l \in \{ 0,\hdots , \mathcal{L} \} } \min_{i \in \mathcal{D}_l} x_i(t_1) \\ 
          &\geq \xMin{t} + \sum_{q=t}^{t_1 -1} \Delta r_{\text{min}}(q) + \min_l \eta_l (x_k(t) - \xMin{t})
    \end{align*} 
    \begin{equation} \label{eq:thm1_xMint1LwrBnd}
        \xMin{t_1} \geq \xMin{t} + \sum_{q=t}^{t_1 -1} \Delta r_{\text{min}}(q) + \eta (x_k(t) - \xMin{t})
    \end{equation} 
    Similarly from \eqref{eq:xi_uprBnd} we get 
    \begin{equation} \label{eq:thm1_xMaxt1UprBnd}
        \xMax{t_1} \leq \xMax{t} + \sum_{q=t}^{t_1 -1} \Delta r_{\text{max}}(q) - \eta (\xMax{t} - x_k(t))
    \end{equation} 
    From \eqref{eq:thm1_xMint1LwrBnd} and \eqref{eq:thm1_xMaxt1UprBnd} we get 
    \begin{equation*} %\label{eq:thm1_yt1}
        y(t_1) = \xMax{t_1} - \xMin{t_1} \leq (1- \eta ) y(t) + \sum_{q=t}^{t_1 -1} \Delta R(q) 
    \end{equation*} 
     With $T_1 = N-1$ we have $t_1 \leq t + T_1$ for all $l$. 
     Now from \eqref{eq:upLaw_vctr} in Lemma \ref{lem:eqModel} and the properties of non-negative entries and row-stochasticity of $W(t)$, we can write for any $i \in \gd$  
     \begin{equation*}
         \xMin{t} + \Delta r_{\text{min}}(t) \leq x_i(t+1) \leq \xMax{t} + \Delta r_{\text{max}}(t)
     \end{equation*} 
     From the above eqn. we can write $y(t+1) \leq y(t) + \Delta R(t)$. Thus we have 
     \begin{equation*}
         y(t+T_1) \leq (1- \eta ) y(t) + \sum_{q=t}^{T_1 -1} \Delta R(q) .
     \end{equation*} 
     Now for a given integer $k \geq 1$, let $T_k = k(N-1)$. Then with $t=1$ we can say 
     \begin{equation*}
         y(T_n + 1) \leq (1- \eta )^n y(1) + \sigma (n) 
     \end{equation*}
    where 
    \begin{equation*}
        \sigma (n) = (1- \eta )^{n-1} \sum_{q=1}^{T_1 -1} \Delta R(q) + \hdots + \sum_{q=T_{n-1}}^{T_n -1} \Delta R(q)
    \end{equation*}
    Let $\lambda$ be the largest integer such that $\lambda(N-1) \leq t$ for any $t \geq 1$. Then, using $y(t+1) \leq y(t) + \Delta R(t)$, we have $\forall t \geq 0$,
    \begin{align*}
        y(t) &\leq y(T_\lambda) + \sum_{q=T_\lambda}^{t-1} \Delta R(q) \\
         &\leq (1- \eta )^\lambda y(1) + \sigma (\lambda) + \sum_{q=T_\lambda}^{t-1} \Delta R(q).
    \end{align*}
    Now as $y(1) \geq 0$ and from the definition of $\lambda$ we have $t/(N-1) - 1 \leq \lambda$, so we can write 
    \begin{equation} \label{eq:thm1_ytEqn}
        y(t) \leq (1- \eta )^{\frac{t}{N-1} - 1} y(1) + \sigma (\lambda) + \sum_{q=T_\lambda}^{t-1} \Delta R(q).
    \end{equation}
    Now we know from Assumption \ref{asmpn:delR_bound} that $\Delta R(t) \leq \theta$. Using this in \eqref{eq:thm1_ytEqn} we can write $y(t) \leq \omega(t)$, where
    \begin{equation*}
        \omega(t) = (1- \eta )^{\frac{t}{N-1} - 1} y(1) + \theta (N-1) (1 + \frac{1 - (1 - \eta)^t}{\eta})
    \end{equation*}
    So in the limiting case we have 
    \begin{equation}
        \lim_{t \rightarrow \infty} y(t) \leq \omega_\infty \text{, where } \omega_\infty = \theta (N-1) (1 + \alpha^{-\frac{1}{2}N(N+1) + 1})
    \end{equation} 
    With the above relation we establish that all the good agents achieve consensus with a bound defined by $\omega$. Next we proceed to show that consensus is achieved over the time-varying average of the reference signals of the trusted agents.
    Let $e(t) := \xBarTd{t} - \rBarTd{t-1}$ where $\xBarTd{t} = \frac{1}{m_T} \one^T \xTd{t}$. Then using the fact that $\one^T \hat{W}(t) = \one^T$ $\forall t$, we get
\begin{align*}
    e(t+1) &= \frac{1}{m_T} \one^T (\hat{W}(t) \xTd{t} + \Delta \rTd{t} - \rTd{t} ) \\ 
     &= \frac{1}{m_T} \one^T (\xTd{t} - \rTd{t-1}) \\
     &= \frac{1}{m_T} \one^T (\hat{W}(t-1) \xTd{t-1} + \rTd{t-2} ) \\ 
     &= \frac{1}{m_T} \one^T (\xTd{t-1} - \rTd{t-2}) \\
     &= \vdots \\ 
     &= \frac{1}{m_T} \one^T (\xTd{1} - \rTd{0}) = e(1)
\end{align*}
Then we can write for all good agents $i \in \gd$
\begin{align*}
    |x_i(t) - \rBarTd{t-1}| &\leq |x_i(t) - \xBarTd{t}| + |e(1)| \\ 
     &\leq |\xMax{t} - \xMin{t}| + |e(1)| \\ 
     &\leq \omega(t) + |e(1)|
\end{align*} 
So in the limiting case we have 
\begin{equation}
    \lim_{t \rightarrow \infty} |x_i(t) - \rBarTd{t-1}| \leq \omega_\infty + |e(1)| = \epsilon
\end{equation} 
\end{proof}

\section{Numerical Simulation} \label{sec:numSim}
In this section we illustrate the effectiveness of our proposed ResDAC algorithm through some numerical simulation results. We consider a network of total 9 agents with the trusted, ordinary and adversarial agents marked in blue, green and red colour respectively. 
\begin{figure}
    \centering
    \begin{subfigure}[b]{0.275\columnwidth}
         \centering
         \includegraphics[scale=0.3]{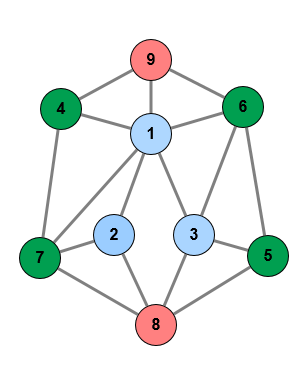}
         \caption{}
         \label{fig:2advA}
     \end{subfigure}
    \hfill
    \begin{subfigure}[b]{0.7\columnwidth}
         \centering
         \includegraphics[scale=0.12]{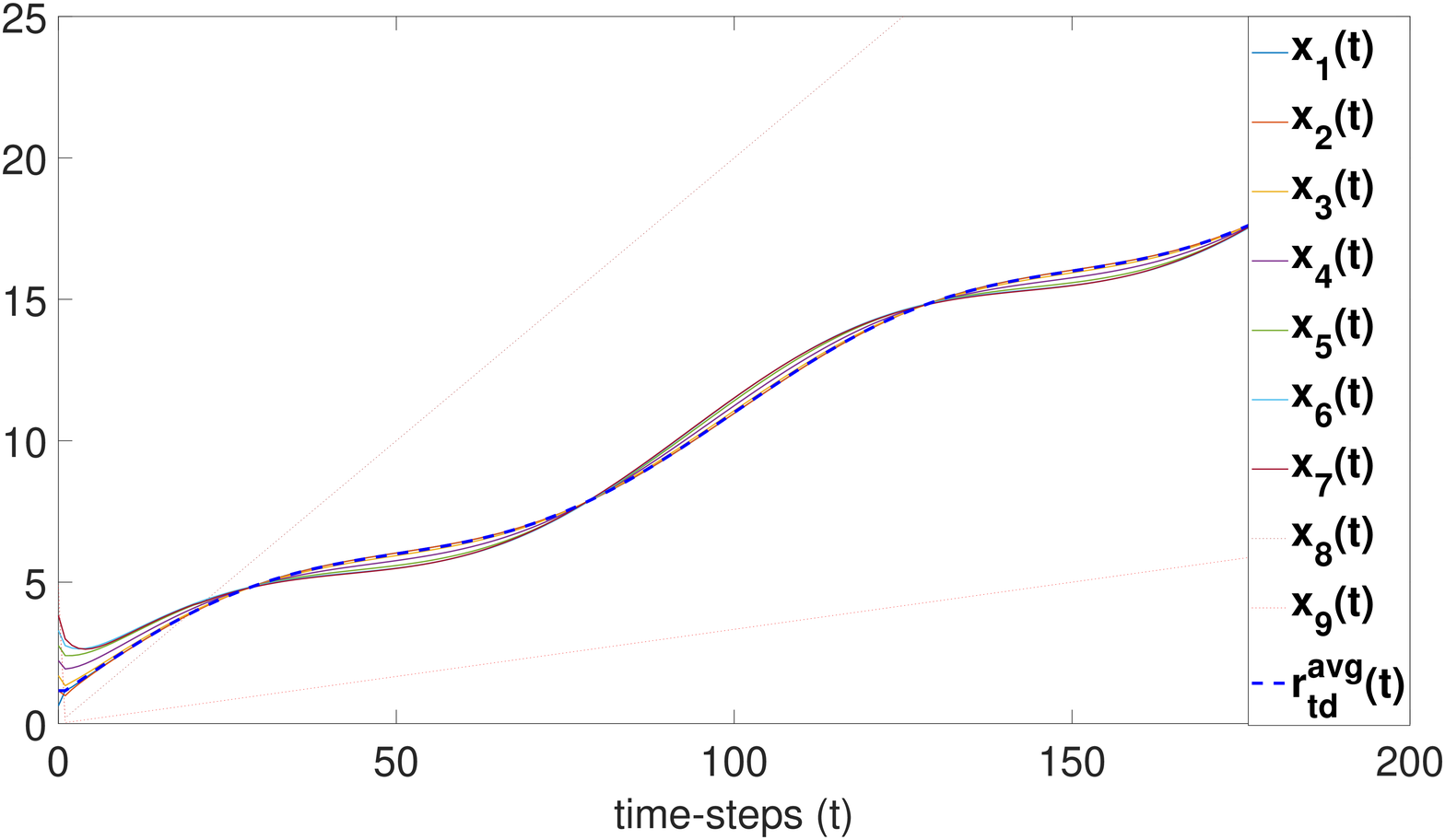}
         \caption{}
         \label{fig:2advB}
     \end{subfigure}
    \caption{(a) Graph with 2 adversarial agents; (b) Performance of ResDAC algorithm in tracking sinusoidal-ramp type signals.}
    \label{fig:2adv}
\end{figure} 
First we consider the network shown in Fig.\ref{fig:2advA} where agents 1-3 are trusted, agents 4-7 are ordinary, and agents 8-9 are adversarial. We consider the reference signals to be a combination of both sinusoidal and ramp type signals as $r_i(t) = 0.5i + t/10 + 0.2isin(0.02 \pi t)$. The adversarial agents try to deviate the estimates far away from an accepted value with $x_8(t) = t/5$ and $x_9(t) = t/30$, shown in Fig.\ref{fig:2advB} by the red dotted lines. Fig.\ref{fig:2advB} clearly shows that following the proposed ResDAC algorithm, the good agents, marked by the solid lines, are able to track the average of the reference signals of the trusted agents, marked by the blue dashed line. The algorithm ensures efficient tracking of the desired value even in the presence of the adversarial agents sharing arbitrary values to its neighbours and trying to disrupt the tracking process.
\begin{figure}
    \centering
    \begin{subfigure}[b]{0.275\columnwidth}
         \centering
         \includegraphics[scale=0.3]{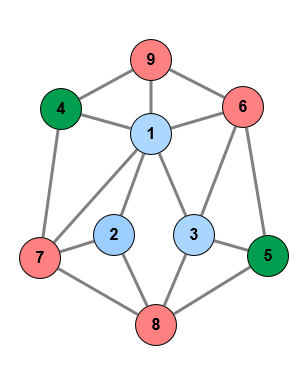}
         \caption{}
         \label{fig:4advA}
     \end{subfigure}
    \hfill
    \begin{subfigure}[b]{0.7\columnwidth}
         \centering
         \includegraphics[scale=0.12]{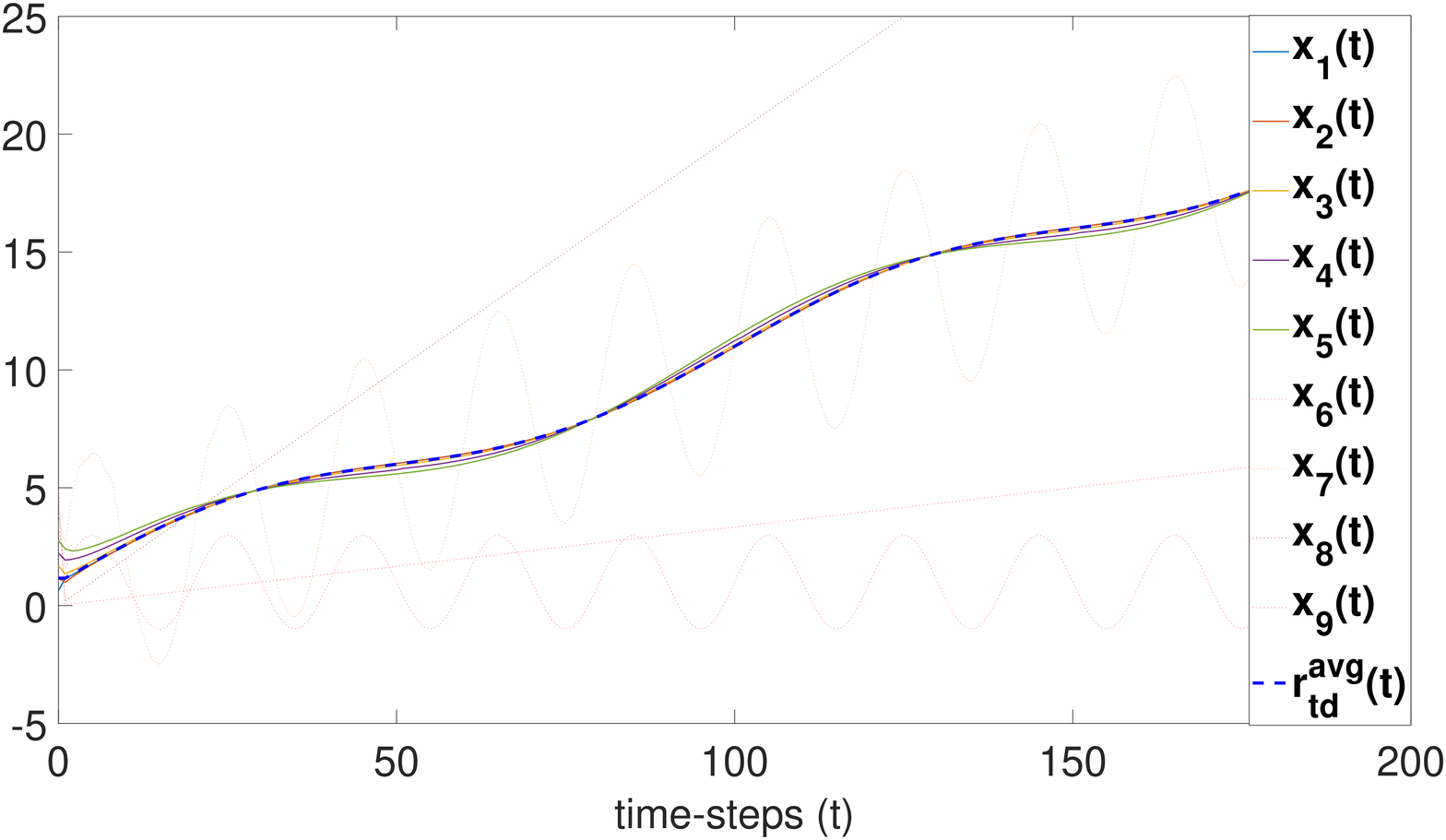}
         \caption{}
         \label{fig:4advB}
     \end{subfigure}
    \caption{(a) Graph with 4 adversarial agents; (b) Performance of ResDAC algorithm in tracking sinusoidal-ramp type signals.}
    \label{fig:4adv}
\end{figure} 

Now to illustrate the fact that the ResDAC algorithm is effective in the presence of any number of adversarial agents as long as the trusted agents induce a CDS, we now consider four adversarial agents in the network compared to two in the previous case. From Fig.\ref{fig:2advA}, consider that two ordinary agents, 6 and 7, have now come under adversarial attack. The new network scenario is shown in Fig.\ref{fig:4advA}. Agents 1-3 remain the trusted agents, and along with the two ordinary agents, 4 and 5, induce a CDS. We consider the same signals as before. Only the two new adversarial agents are assigned as $x_6(t) = 1 + 2sin(0.1 \pi t), x_7(t) = 3.5 + t/10 + 5sin(0.1 \pi t)$. Fig.\ref{fig:4advB} clearly shows that in the presence of different arbitrary inputs by the four adversarial agents, the ordinary and trusted agents are able to efficiently track the desired value following the ResDAC algorithm. Note that the ordinary agents, 4 and 5, are connected to only one trusted neighbour, 1 and 3 respectively. All the other neighbours for the ordinary agents are adversarial. 

\section{Conclusion}
In this paper, we define resilient DAC based on trusted agents in the presence of Byzantine and malicious adversaries. Then we develop the novel ResDAC algorithm to ensure that the trusted and ordinary agents are able to achieve resilient DAC in the presence of adversarial attacks. We show that when the trusted agents induce a CDS, and the reference signals are slowly varying, all the good agents are able to track the desired time-varying average value by following the ResDAC algorithm. Through numerical simulations we show the effectiveness of the ResDAC algorithm even in the presence of large number of adversaries within the network. Future direction of work is to consider other modes of adversarial attacks. 

\section*{APPENDIX}

\subsection{Proof of Theorem \ref{thm:consensus}} \label{app:proof_consensus} 

\begin{proof}
The ResDAC algorithm has two distinct parts - one for the state update of the trusted agents, and the other for the ordinary agents. Let us first consider the case for the trusted agents' state update presented in update law \eqref{eq:upLaw_td}. 
 A trusted agent $i \in \node_T$ uses states from only its trusted neighbours and so $\Hat{W}_{ij}(t) =  W_{ij}(t) = 1/m_T$ for all $(j,i) \in \edge_G$. The diagonal element, $\Hat{W}_{ii}(t) =  W_{ii}(t) = 1 - |\mathcal{T}_i|/m_T$. So $W_{ij}(t) \neq 0$ and $\Hat{W}_{ij}(t) \neq 0$ for all $j \in \edge_G \cup \{(i,i)\}$, which satisfies property 1B. As $|\mathcal{T}_i| \leq m_T - 1$, $\Hat{W}_{ii}(t) =  W_{ii}(t) \geq 1/m_T$. So property 1C is also satisfied. Now for the $i$-th row of $W(t)$, 
 \begin{equation*}
     \sum_{j=1}^N W_{ij}(t) = \sum_{j \in \mathcal{T}_i} W_{ij}(t) + W_{ii}(t) = 1.
 \end{equation*} 
 The above equation satisfies 1A. Now using this and the fact that the graph is undirected we establish 2A as 
 \begin{equation*}
     \sum_{k=1}^{m_T} \Hat{W}_{ki}(t) = \sum_{j=1}^{m_T} \Hat{W}_{ij}(t) = \sum_{j \in \mathcal{T}_i} W_{ij}(t) + W_{ii}(t) = 1.
 \end{equation*} 
 So we have established all the properties of the sub-matrix $\Hat{W}(t)$, and for the first $m_T$ rows of matrix $W(t)$. Now we proceed to establish the properties for the remaining rows of matrix $W(t)$ by considering the case of the state update of the ordinary agents. Consider an ordinary agent $i \in \node_G$. It utilizes the states of only those neighbours $x_j(t), j \in \mathcal{U}_i(t)$ that lie within the maximum and minimum states among its neighbouring trusted nodes and itself. So we have $|\mathcal{U}_i(t)| \leq |\mathcal{N}_i|$, which further means $1/|\mathcal{U}_i(t)| \geq 1/(|\mathcal{N}_i| + 1)$. Let us divide the states being used by agent $i$ for the update at iteration $t$ into three sets : 
\begin{itemize}
    \item $\nbrTdN{i}{t}$ : the set of its trusted neighbours, 
    \item $\{ i \}$ : the set containing the agent itself, 
    \item $\nbrAdv{i}{t}$ : the set of its ordinary and adversarial neighbours.
\end{itemize}
Then we rewrite the first part of the state update in \eqref{eq:upLaw_ord} as 
\begin{equation} \label{eq:lem1_stUpdt3prts}
    \sum_{j \in \mathcal{U}_i(t)} \frac{x_j(t)}{|\mathcal{U}_i(t)|} = \sum_{j \in \nbrTdN{i}{t}} \frac{x_j(t)}{|\mathcal{U}_i(t)|} + \frac{x_i(t)}{|\mathcal{U}_i(t)|} + \sum_{j \in \nbrAdv{i}{t}} \frac{x_j(t)}{|\mathcal{U}_i(t)|}
\end{equation}
Now we look at two separate cases based on the set $\nbrAdv{i}{t}$. 

Case I [$\nbrAdv{i}{t} = \phi$] : $W_{ij}(t) = 1/|\mathcal{U}_i(t)|$ holds for all $(j,i) \in \edge_G$ and $j=i$. This satisfies all the conditions in Lemma \ref{lem:eqModel}.

Case II [$\nbrAdv{i}{t} \neq \phi$] : this means at least one ordinary or adversarial agent exists in $\mathcal{U}_i(t)$. Let us consider the case when there is only one ordinary or adversarial neighbour $k \in \mathcal{U}_i(t)$. Then from RDC algorithm we can say that $x_i^{\text{min}}(t) \leq x_k(t) \leq x_i^{\text{max}}(t)$. So there exists $\lambda \in [0,1]$ such that $x_k(t) = \lambda x_i^{\text{max}}(t) + (1- \lambda) x_i^{\text{min}}(t)$. Let us also consider that only two distinct agents $j_a$ and $j_b$ exists in $\nbrTdN{i}{t} \cup \{ i \}$ such that $x_{j_a}(t) = x_i^{\text{max}}(t)$ and $x_{j_b}(t) = x_i^{\text{min}}(t)$. Let $\mathcal{U}'_i(t) := \mathcal{U}_i(t) \backslash \{ j_a \cup j_b \cup k \} $. 
Then from \eqref{eq:lem1_stUpdt3prts} we have 
\begin{align*}
    \sum_{j \in \mathcal{U}_i(t)} \frac{x_j(t)}{|\mathcal{U}_i(t)|} &= \sum_{j \in \nbrTdN{i}{t}} \frac{x_j(t)}{|\mathcal{U}_i(t)|} + \frac{x_i(t)}{|\mathcal{U}_i(t)|} + \frac{x_k(t)}{|\mathcal{U}_i(t)|} \\ 
     = \sum_{j \in \nbrTdN{i}{t}} \frac{x_j(t)}{|\mathcal{U}_i(t)|} &+ \frac{x_i(t)}{|\mathcal{U}_i(t)|} + \frac{\lambda x_i^{\text{max}}(t) + (1- \lambda) x_i^{\text{min}}(t)}{|\mathcal{U}_i(t)|} \\ 
     = \sum_{j \in \mathcal{U}'_i(t)} \frac{x_j(t)}{|\mathcal{U}_i(t)|} &+ \frac{(1+\lambda) x_{j_a}(t) + (2- \lambda) x_{j_b}(t)}{|\mathcal{U}_i(t)|} 
\end{align*}
Thus we have 
\begin{align*}
    W_{ij_a}(t) = \frac{1+\lambda}{|\mathcal{U}_i(t)|} &, W_{ij_b}(t) = \frac{2-\lambda}{|\mathcal{U}_i(t)|}, \\ 
    \text{ and } \forall j \in \mathcal{U}'_i(t), W_{ij}(t) &= \frac{1}{|\mathcal{U}_i(t)|}.
\end{align*}
 This satisfies the properties 1B and 1C of Lemma \ref{lem:eqModel}. Now, $|\mathcal{U}'_i(t)| + (1+\lambda) + (2-\lambda) = |\mathcal{U}'_i(t)| + 3$, which is equal to $|\mathcal{U}_i(t)|$ from the definition of $\mathcal{U}'_i(t)$. This satisfies property 1A of row-stochasticity. 
 Note that in case there are more than one agent in $\nbrTdN{i}{t} \cup \{ i \}$ with their state values equal to the maximum value $x_i^{\text{max}}(t)$, then the agent $j_a$ can be chosen arbitrarily from among them, and the above analysis would still hold true. Similarly for the case of more than one agent in $\nbrTdN{i}{t} \cup \{ i \}$ with their state values equal to the minimum value $x_i^{\text{min}}(t)$.
 Now in case of two or more ordinary or adversarial agents in $\nbrAdv{i}{t}$
 Proceeding similarly for the case of , or multiple nodes with maximum or minimum values in $\mathcal{U}_i(t)$, we will arrive at the same result showing the existence of a suitable matrix $W(t)$ satisfying the required conditions. 
\end{proof}

\subsection{Proof of Proposition \ref{prop:xMax_xMin}} \label{app:proof_propn} 
\begin{proof}
    Without loss of generality, consider $s \geq 0$. 
\end{proof}

\printbibliography

@article{ResDO_Trstd_TAC20, 
  author={Zhao, Chengcheng and He, Jianping and Wang, Qing-Guo},
  journal={IEEE Transactions on Automatic Control}, 
  title={Resilient Distributed Optimization Algorithm Against Adversarial Attacks}, 
  year={2020},
  volume={65},
  number={10},
  pages={4308-4315}
  %doi={10.1109/TAC.2019.2954363}
}

@article{DT_DAC,
title = {Discrete-time dynamic average consensus},
journal = {Automatica},
volume = {46},
number = {2},
pages = {322-329},
year = {2010},
%issn = {0005-1098},
%doi = {https://doi.org/10.1016/j.automatica.2009.10.021},
author = {Minghui Zhu and Sonia Martínez}
%keywords = {Consensus algorithms, Multi-agent systems, Cooperative control},
}

@ARTICLE{DAC_Mag19,
  author={Kia, Solmaz S. and Van Scoy, Bryan and Cortes, Jorge and Freeman, Randy A. and Lynch, Kevin M. and Martinez, Sonia},
  journal={IEEE Control Systems Magazine}, 
  title={Tutorial on Dynamic Average Consensus: The Problem, Its Applications, and the Algorithms}, 
  year={2019},
  volume={39},
  number={3},
  pages={40-72} 
}

@article{ResCon_Sundaram_13,
title = {Resilient Asymptotic Consensus in Robust Networks}, 
author = {LeBlanc, Heath J. and Zhang, Haotian and Koutsoukos, Xenofon and Sundaram, Shreyas}, 
journal = {IEEE Journal on Selected Areas in Communications}, 
year={2013},
volume={31},
number={4},
pages={766-781}
}

@article{ResDAC,
title = {Resilient Dynamic Average-Consensus of Multiagent Systems}, 
author = {Iqbal, Muhammad and Qu, Zhihua and Gusrialdi, Azwirman}, 
journal = {IEEE Control Systems Letters},
year = {2022},
volume = {6}, 
pages = {3487-3492}
}

@ARTICLE{ResCon_Ishii_20,
  author={Wang, Yuan and Ishii, Hideaki},
  journal={IEEE Transactions on Control of Network Systems}, 
  title={Resilient Consensus Through Event-Based Communication}, 
  year={2020},
  volume={7},
  number={1},
  pages={471-482}
  %doi={10.1109/TCNS.2019.2924235}
  }

@ARTICLE{DAC_Murray_IFAC05,
  author={Spanos, Demetri P. and Olfati-Saber, Reza and Murray, Richard M.},
  journal={Proceedings of IFAC World Congress}, 
  title={Dynamic Consensus for Mobile Networks}, 
  year={2005},
  pages={1-6},
  %doi={10.1109/TAC.2019.2954363}
  }

@ARTICLE{REWB,
  author={Bhattacharyya, Shamik and Rokade, Kiran and Kalaimani, Rachel Kalpana},
  journal={IEEE Transactions on Control of Network Systems}, 
  title={Distributed Estimation over Directed Graphs Resilient to Sensor Spoofing}, 
  year={2023},
  volume={},
  number={},
  pages={1-11}
  %doi={10.1109/TCNS.2023.3247560}
  }

@ARTICLE{AttckSrvey_IndInfo22,
  author={He, Wangli and Xu, Wenying and Ge, Xiaohua and Han, Qing-Long and Du, Wenli and Qian, Feng},
  journal={IEEE Transactions on Industrial Informatics}, 
  title={Secure Control of Multiagent Systems Against Malicious Attacks: A Brief Survey}, 
  year={2022},
  volume={18},
  number={6},
  pages={3595-3608}
  %doi={10.1109/TII.2021.3126644}
  }

@INPROCEEDINGS{ResCon_Trstd_14,
  author={Abbas, Waseem and Vorobeychik, Yevgeniy and Koutsoukos, Xenofon},
  booktitle={2014 7th International Symposium on Resilient Control Systems (ISRCS)}, 
  title={Resilient consensus protocol in the presence of trusted nodes}, 
  year={2014},
  volume={},
  number={},
  pages={1-7}
  %doi={10.1109/ISRCS.2014.6900100}
  }

@ARTICLE{TdNodes_Xeno_TCNS18,
  author={Abbas, Waseem and Laszka, Aron and Koutsoukos, Xenofon},
  journal={IEEE Transactions on Control of Network Systems}, 
  title={Improving Network Connectivity and Robustness Using Trusted Nodes With Application to Resilient Consensus}, 
  year={2018},
  volume={5},
  number={4},
  pages={2036-2048}
  %doi={10.1109/TAC.2008.2009515}
  }

@ARTICLE{ResCon_DoS,
  author={Zuo, Zhiqiang and Cao, Xiong and Wang, Yijing and Zhang, Wentao},
  journal={IEEE Transactions on Systems, Man, and Cybernetics: Systems}, 
  title={Resilient Consensus of Multiagent Systems Against Denial-of-Service Attacks}, 
  year={2022},
  volume={52},
  number={4},
  pages={2664-2675},
  % doi={10.1109/TSMC.2021.3051730} 
  }

\end{document}